\documentclass[a4paper, onecolumn, 11pt, accepted=2026-05-05]{quantumarticle}
\pdfoutput=1
\usepackage[utf8]{inputenc}
\usepackage[english]{babel}
\usepackage[T1]{fontenc}
\usepackage{amsmath, amssymb, amsfonts, amsthm}
\usepackage{hyperref}

\usepackage{tikz}
\usetikzlibrary{positioning}
\usetikzlibrary{shapes}
\usetikzlibrary{decorations.pathreplacing}
\usetikzlibrary{calc,patterns,angles,quotes}
\usetikzlibrary{intersections}
\usetikzlibrary{hobby}

\newtheorem{theorem}{Theorem}[section]

\newtheorem{remark}[theorem]{Remark}
\newtheorem{lemma}[theorem]{Lemma}
\newtheorem{claim}[theorem]{Claim}

\newtheorem{definition}[theorem]{Definition}

\newcommand{\ket}[1]{|#1\rangle}

\newcommand{\ketbra}[2]{|#1\rangle\langle#2|}
\newcommand{\braket}[2]{\langle#1|#2\rangle}
\newcommand{\norm}[1]{\left\lVert#1\right\rVert}

\newcommand{\C}{\mathbb{C}}
\newcommand{\eps}{\varepsilon}
\renewcommand{\epsilon}{\varepsilon}

\newcommand{\cbra}[1]{\left\{#1\right\}}
\newcommand{\rbra}[1]{\left(#1\right)}
\newcommand{\mathify}[1]{\ifmmode{#1}\else\mbox{$#1$}\fi}

\newcommand{\zone}{\{0, 1\}}

\newcommand{\for}{\mathsf{frOR}}

\makeatletter
\let\oldabs\abs
\def\abs{\@ifstar{\oldabs}{\oldabs*}}
\makeatother

\newcommand{\OR}{\mathsf{OR}}

\newcommand{\QPE}{\mathsf{QPE}}
\newcommand{\maxqpe}{\mathsf{maxQPE}}

\newcommand{\dist}{\mathsf{dist}}

\newcommand{\A}{\mathcal{A}}

\newcommand{\R}{\mathbb{R}}

\newcommand{\tO}{\widetilde{O}}

\usepackage[numbers,sort&compress]{natbib}

\begin{document}

\title{Tight Bounds for Quantum Phase Estimation and Related Problems}

\author{Nikhil S.~Mande}
\affiliation{University of Liverpool, UK}
\orcid{0000-0002-9520-7340}
\email{mande@liverpool.ac.uk}
\homepage{https://mande-nikhil.github.io/}
\thanks{Part of this work was done while the author was a postdoc at QuSoft and CWI, Amsterdam.}
\author{Ronald de Wolf}
\affiliation{QuSoft, CWI and University of Amsterdam, the Netherlands}
\email{rdewolf@cwi.nl}
\homepage{https://homepages.cwi.nl/~rdewolf/}
\thanks{Partially supported by the Dutch Research Council (NWO/OCW), as part of the Quantum Software Consortium programme (project number 024.003.037).}
\maketitle

\begin{abstract}
    Phase estimation, due to Kitaev [arXiv'95], is one of the most fundamental subroutines in quantum computing, used in Shor's factoring algorithm, optimization algorithms, quantum chemistry algorithms, and many others. In the basic scenario, one is given black-box access to a unitary~$U$, and an eigenstate $\ket{\psi}$ of $U$ with unknown eigenvalue $e^{i\theta}$, and the task is to estimate the eigenphase~$\theta$ within $\pm\delta$, with high probability. The repeated application of $U$ and $U^{-1}$ is typically the most expensive part of phase estimation algorithms, so for us the \emph{cost} of an algorithm will be that number of applications.

    Motivated by the ``guided local Hamiltonian problem'' from quantum chemistry, we tightly characterize the cost of several variants of phase estimation where we are no longer given an arbitrary eigenstate, but are required to estimate the \emph{maximum} eigenphase of $U$, aided by \emph{advice} in the form of copies of a state (or a unitary preparing that state) that is promised to have at least a certain overlap~$\gamma$ with the top eigenspace. 
    We give algorithms and matching lower bounds (up to logarithmic factors) for all ranges of parameters.
    We show a crossover point below which advice is not helpful: $o(1/\gamma^2)$ copies of the advice state (or $o(1/\gamma)$ applications of an advice-preparing unitary) are not significantly better than having no advice at all. We also show that having much more advice (more than $O(1/\gamma^2)$ copies or more than  $O(1/\gamma)$ applications of the advice-preparing unitary) does not  significantly reduce cost, and neither does knowledge of the eigenbasis of $U$. 
As an immediate consequence of a key technical component of our proof, we obtain a lower bound on the complexity of the Unitary recurrence time problem, matching an upper bound of She and Yuen~[ITCS'23] and resolving one of their open questions.
    
    Lastly, we study how efficiently one can reduce the error probability in the basic phase-estimation scenario. We show that a phase-estimation algorithm with precision $\delta$ and error probability $\eps$ has cost $\Omega\left(\frac{1}{\delta}\log\frac{1}{\epsilon}\right)$, matching the obvious way to amplify the basic constant-error-probability phase estimation algorithm. This contrasts with some other scenarios in quantum computing (e.g., search) where error-probability reduction costs only a factor $O(\sqrt{\log(1/\epsilon)})$. Our lower bound uses a variant of the polynomial method with  trigonometric polynomials.
\end{abstract}

\thispagestyle{empty}
\addtocounter{page}{-1}

\newpage

\section{Introduction}

\subsection{Phase estimation}

Kitaev~\cite{Kit95} gave an elegant and efficient quantum algorithm for the task of \emph{phase estimation} nearly 30 years ago. The task is easy to state: given black-box access to a unitary and an eigenstate of it, estimate the phase of the associated eigenvalue. 
Roughly speaking, the standard algorithm for this task sets up a superposition involving many different powers of the unitary to extract many different powers of the eigenvalue, and then uses a quantum Fourier transform to turn that into an estimate of the eigenphase.\footnote{An added advantage of the standard algorithm for phase estimation is that it can also work with a quantum Fourier transform that is correct on average rather than in the worst case~\cite{linden&wolf:avQFT}. However, there are also approaches to phase estimation that avoid the QFT altogether, see e.g.~\cite{rall:phasest}.}
Many of the most prominent quantum algorithms can either be phrased as phase estimation, or use phase estimation as a crucial subroutine. Some examples are Shor's period-finding algorithm~\cite{shor:factoring} as presented in~\cite{cemm:revisited}; approximate counting~\cite{bhmt:countingj} can be done using phase estimation on the unitary of one iteration of Grover's search algorithm~\cite{grover:search}, which also recovers the $O(\sqrt{N})$ complexity for searching an $N$-element unordered search space; the HHL algorithm for solving linear systems of equations estimates eigenvalues in order to invert them~\cite{hhl:lineq}. Applications of phase estimation in quantum chemistry are also very prominent, as discussed below.

More precisely, in the task of phase estimation, we are given black-box access to an $N$-dimensional unitary~$U$ (and a controlled version thereof) and a state $\ket{\psi}$ that satisfies $U\ket{\psi} = e^{i \theta}\ket{\psi}$. Our goal is to output (with probability at least $2/3$) a $\Tilde{\theta} \in [0, 2\pi)$ such that $|\tilde{\theta} - \theta|$ is at most $\delta$ in $\mathbb{R}$~mod $2\pi$. In the basic scenario we are given access to one copy of $\ket{\psi}$, and are allowed to apply $U$ and its inverse. Since the repeated applications of $U$ and $U^{-1}$ are typically the most expensive parts of algorithms for phase estimation, the \emph{cost} we wish to minimize is the number of applications of $U$ and $U^{-1}$. 
We are additionally allowed to apply arbitrary unitaries that do not depend on $U$, at no cost.
Kitaev's algorithm has cost $O(1/\delta)$ in this measure.

\subsection{Phase estimation with advice}

One of the core problems in quantum chemistry is the following: given a classical description of some Hamiltonian $H$ (for instance an ``electronic structure'' Hamiltonian in the form of a small number of local terms), estimate its \emph{ground-state energy}, which is its smallest eigenvalue. If $H$ is normalized such that its eigenvalues are all in $[0,2\pi)$ and we define the unitary $U=e^{iH}$ (which has the same eigenvectors as $H$, with an eigenvalue $\lambda$ of $H$ becoming the eigenvalue $e^{i\lambda}$ for $U$), then finding the ground state energy of $H$ is equivalent to finding the smallest eigenphase of $U$. If we are additionally given a \emph{ground state}~$\ket{\psi}$ (i.e., an eigenstate corresponding to the smallest eigenphase), then phase estimation is tailor-made to estimate the ground-state energy. However, in quantum chemistry it is typically hard to prepare the ground state of $H$, or even something close to it. What can sometimes be done is a relatively cheap preparation of some quantum ``advice state'' that has some non-negligible ``overlap''~$\gamma$ with the ground space, for instance the ``Hartree-Fock state''.
This idea is also behind the Variational Quantum Eigensolver (VQE)~\cite{VQE}.
In the complexity-theoretic context, this problem of ground-state energy estimation for a local Hamiltonian given an advice state, is known as the 
``guided local Hamiltonian problem'', and has received quite some attention recently~\cite{GL23,CFG+23,WFC24} because of its connections with quantum chemistry as well as deep complexity questions such as the PCP conjecture.
These complexity-theoretic results typically focus on completeness for BQP or QMA or QCMA of certain special cases of the guided local Hamiltonian problem, and don't care about polynomial overheads of the cost 
(number of uses of $U$ and $U^{-1}$) 
in the number of qubits $\log N$ or in the parameters $\delta$ and $\gamma$. In contrast, we care here about getting essentially optimal bounds on the cost of phase estimation in various scenarios.

To be more precise, suppose our input unitary is $U = \sum_{j = 0}^{N-1} e^{i\theta_j}\ketbra{u_j}{u_j}$ with each $\theta_j \in [0, 2\pi-2\delta)$.\footnote{\label{note:hannote}By multiplying $U$ with $e^{i\delta}$, we can equivalently assume that all eigenphases are in the interval $[\delta,2\pi-\delta)$, hence $\not\in[-\delta,\delta] \mod 2\pi$. In an earlier version of this paper we did not include this promise in the problem definition. 
However, then the algorithm of Lemma~\ref{lemma: maxthetafind} runs into trouble: if $\theta_{\max}$ is, say, $\pi$, while all other eigenphases are just above~0, then the unitary $V$ in that proof generates a state with non-negligible weight on values $x_k$ that are just below $2\pi$ (since those are good approximations of the just-above-0 eigenphases mod $2\pi$), and the generalized maximum-finding will zoom in on those values, ending up with some number just below~$2\pi$ as its (very wrong) guess for $\theta_{\max}$. We thank Han-Hsuan Lin (personal communication) for alerting us to this issue.}  Let $\theta_{\max} = \max_{j \in \cbra{0, 1, \dots, N-1}}\theta_j$ denote the maximum eigenphase, and let $S$ denote the space spanned by all eigenstates with eigenphase $\theta_{\max}$, i.e., the ``top eigenspace''. Advice is given in the form of a state $\ket{\alpha}$ whose projection on $S$ has squared norm at least $\gamma^2$: $\norm{P_S\ket{\alpha}}^2 \geq \gamma^2$. Note that if $S$ is spanned by a single eigenstate~$\ket{u_{\max}}$, then this condition is the same as $|\braket{\alpha}{u_{\max}}| \geq \gamma$, which is why we call $\gamma$ the \emph{overlap} of the advice state with the target eigenspace. The task $\maxqpe_{N, \delta}$ is to output, with probability at least $2/3$, a $\delta$-precise (in $\R$ mod $2\pi$) estimate of~$\theta_{\max}$.\footnote{It doesn't really matter, but we focus on the \emph{maximum} rather than minimum eigenphase of~$U$ because eigenphase~0 (i.e., eigenvalue~1) is a natural baseline, and we are looking for the eigenphase furthest away from this baseline.}

We will distinguish between the setting where the advice is given in the form of a number of copies of the advice state $\ket{\alpha}$, or the potentially more powerful setting where we can apply (multiple times) a \emph{unitary}~$A$ that prepares $\ket{\alpha}$ from some free-to-prepare initial state, say $\ket{0}$. We would have such a unitary~$A$ for instance if we have a procedure to prepare $\ket{\alpha}$ ourselves in the lab.
We can also distinguish between the situation where the eigenbasis $\ket{u_0},\ldots,\ket{u_N}$ of $U$ is known (say, the computational basis where $\ket{u_j}=\ket{j}$) and the potentially harder situation where the eigenbasis is unknown.
These two binary distinctions give us four different settings.
For each of these settings we determine essentially optimal bounds on the cost of phase estimation, summarized in Table~\ref{table: bounds} 
(see Section~\ref{sec: prelims} for the formal setup). 

\begin{table}[hbt]
\begin{center}
\resizebox{\columnwidth}{!}{%
\begin{tabular}{| c | c | c | c | c | c |}
 \hline
 Row & Basis & Access to advice & Number of accesses & Upper bound & Lower bound \\ [0.5ex] 
 \hline\hline
 1 & known & state & $o\rbra{\frac{1}{\gamma^2}}$ & $\tO\rbra{\frac{\sqrt{N}}{\delta}}$, Lemma~\ref{lemma: odd ub} & $\Omega\rbra{\frac{\sqrt{N}}{\delta}}$, Lemma~\ref{lemma:13lb} \\
 \hline
 2 & known & state & $\Omega\rbra{\frac{1}{\gamma^2}}$ & $\tO\rbra{\frac{1}{\gamma\delta}}$, Lemma~\ref{lemma: 24 ub} & $\Omega\rbra{\frac{1}{\gamma\delta}}$, Lemma~\ref{lemma:24lb} \\
 \hline
 3 & unknown & state & $o\rbra{\frac{1}{\gamma^2}}$ & $\tO\rbra{\frac{\sqrt{N}}{\delta}}$, Lemma~\ref{lemma: odd ub} & $\Omega\rbra{\frac{\sqrt{N}}{\delta}}$, Lemma~\ref{lemma:13lb} \\
 \hline
 4 & unknown & state & $\Omega\rbra{\frac{1}{\gamma^2}}$ & $\tO\rbra{\frac{1}{\gamma\delta}}$, Lemma~\ref{lemma: 24 ub} & $\Omega\rbra{\frac{1}{\gamma\delta}}$, Lemma~\ref{lemma:24lb} \\ 
 \hline
 5 & known & unitary & $o\rbra{\frac{1}{\gamma}}$ & $\tO\rbra{\frac{\sqrt{N}}{\delta}}$, Lemma~\ref{lemma: odd ub} & $\Omega\rbra{\frac{\sqrt{N}}{\delta}}$, Lemma~\ref{lemma:57lb} \\
 \hline
 6 & known & unitary & $\Omega\rbra{\frac{1}{\gamma}}$ & $\tO\rbra{\frac{1}{\gamma\delta}}$, Lemma~\ref{lemma: 68 ub} & $\Omega\rbra{\frac{1}{\gamma\delta}}$, Lemma~\ref{lemma:68lb} \\
 \hline
 7 & unknown & unitary & $o\rbra{\frac{1}{\gamma}}$ & $\tO\rbra{\frac{\sqrt{N}}{\delta}}$, Lemma~\ref{lemma: odd ub} & $\Omega\rbra{\frac{\sqrt{N}}{\delta}}$, Lemma~\ref{lemma:57lb} \\
 \hline
 8 & unknown & unitary & $\Omega\rbra{\frac{1}{\gamma}}$ & $\tO\rbra{\frac{1}{\gamma\delta}}$, Lemma~\ref{lemma: 68 ub} & $\Omega\rbra{\frac{1}{\gamma\delta}}$, Lemma~\ref{lemma:68lb} \\
 \hline
\end{tabular}
}
\caption{\label{table: bounds}
Our results for the cost of $\maxqpe_{N, \delta}$.
We assume $\gamma \geq 1/\sqrt{N}$ since a random state has overlap $1/\sqrt{N}$ with the target eigenspace with high probability, and such a state can be prepared at no cost. The `Basis' column indicates whether the eigenbasis of $U$ is known; `Access to advice' indicates whether we get copies of the advice state or a unitary to prepare it;  `Number of accesses' refers to the number of accesses to advice that we have. The last two columns show our bounds  with references to the lemmas where they are stated and proved. The $\tO(\cdot)$ in the upper-bound column hides a factor $\log N$ for the odd-numbered rows, and $\log(1/\gamma)$ for the even-numbered rows. The lower bounds assume $\delta\in(0,1)$.}
\end{center}
\end{table}

Let us highlight some interesting consequences of our results. First, a little bit of advice is no better than no advice at all: the upper bounds in the odd-numbered rows of Table~\ref{table: bounds} are actually obtained by algorithms that don't use the given advice ($o(1/\gamma^2)$ copies of $\ket{\alpha}$ or $o(1/\gamma)$ applications of $A$ and $A^{-1}$) at all,  yet their costs essentially match our lower bounds for algorithms that use the given advice.%
\footnote{The proofs in Section~\ref{sec: lower} 
yield the same asymptotic lower bounds for algorithms with access to at most $c/\gamma^2$ advice states for Theorem~\ref{thm: gndeltat lower bound}, Rows 1 and 3 of Table~\ref{table: bounds}, and for algorithms with access to at most $c/\gamma$ advice unitaries for Rows 5 and 7 of Table~\ref{table: bounds}, where $c$ is a suitably small 
universal 
constant. We chose to use $o(\cdot)$ here to avoid clutter.}

A second interesting consequence is that too much advice is no better than a moderate amount of advice: the upper bounds in Rows 2 and~4 use $O(1/\gamma^2)$ advice states, and the upper bounds in Rows 6 and 8 use $O(1/\gamma)$ advice unitaries, and using more advice does not reduce the cost further.

Thirdly, it turns out that knowledge of the eigenbasis of $U$ doesn't really help in reducing the cost: the costs in row~1 and row~3 are the same, and similarly for rows 2 vs.~4, 5 vs.~7 and 6 vs.~8.

\paragraph*{Our techniques.}
Our upper bounds use the subroutine of \emph{generalized maximum-finding} of van Apeldoorn, Gily{\'{e}}n, Gribling, and de Wolf~\cite{AGGW20} which allows us to find maximum values in the second register of a two-register superposition even when the first of these two registers has an unknown basis. We derive the upper bound of row~4 from the upper bound of row~8 by using roughly $1/\gamma$ copies of $\ket{\alpha}$ to simulate one reflection around the state~$\ket{\alpha} = A\ket{0}$, using the techniques of Lloyd, Mohseni, and Rebentrost~\cite{lmr:pca}.

Our lower bounds follow via reductions from a fractional version of the Boolean OR function with advice. We show a lower bound for this by a simple modification of the adversary method~\cite{Amb02}, taking into account the input-dependent advice in the initial state, and the fact that applications of $U$ can be made to correspond to ``fractional'' queries.

\paragraph*{Gate-complexity of our algorithms.}
We stated the \emph{cost} (i.e., number of applications of $U$ and $U^{-1}$) of our algorithms here in the upper-bound column of Table~\ref{table: bounds}, not the overall time complexity. However, it is easy to verify that the gate-complexities of our algorithms are only larger than the cost by log-factors in all cases except rows 2 and~4: our algorithms use two main subroutines, which have only small overheads in gate-complexity, namely basic quantum phase estimation~\cite{Kit95} and generalized maximum-finding~\cite{AGGW20}. In contrast, our upper bound for row~4 (and hence for row~2) additionally uses \cite{lmr:pca} $O(1/\gamma)$ times to implement a reflection about the state $\ket{\alpha}$, consuming $O(1/\gamma)$ copies of that state for each reflection. One such reflection then has cost~0 but gate-complexity $\tO(1/\gamma)$, meaning the overall gate-complexity of our algorithm for row~4 is $\tO(1/\gamma\delta + 1/\gamma^2)$ rather than $\tO(1/\gamma\delta)$, which makes a difference if $\gamma\ll\delta$.

\paragraph*{Comparison with related work.}
Some of the results in our table were already (partially) known.
A cost-$\tO(\sqrt{N}/\delta)$ algorithm for the adviceless setting with unknown eigenbasis (implying the upper bounds of rows 1, 3, 5, 7) was originally due to Poulin and Wocjan~\cite{poulin2009GibbsSamplingAndEval}, and subsequently improved in the log-factors by van Apeldoorn et al.~\cite{AGGW20};
the latter algorithm is basically our proof of Lemma~\ref{lemma: odd ub}.
Lin and Tong~\cite{linlin&tong:groundstateprep} (improving upon~\cite{GTC:groundstateprep}) studied the situation with an advice-preparing unitary. Their setting is slightly different from ours, they focus on preparing the ground state\footnote{Because generalized maximum-finding (Lemma~\ref{lemma: genmaxfind}) actually outputs a state in addition to an estimate, our algorithms can be modified to also output a state that is close to the top eigenspace of~$U$.}
 of a given Hamiltonian without a known bound on its spectrum, but~\cite[Theorem~8]{linlin&tong:groundstateprep} implies  
a cost-$O(\log(1/\gamma)\log(1/\delta)\log\log(1/\delta)/\gamma\delta)$ algorithm for our row~8.
Their follow-up paper~\cite{linlin&tong:grounstateearly} further reduces the number of auxiliary qubits with a view to near-term implementation, but does not reduce the cost further.
Our cost-$O(\log(1/\gamma)/\gamma\delta)$ algorithm is slightly better in the log-factors than theirs, and uses quite different techniques (\cite{linlin&tong:groundstateprep} uses quantum singular value transformation~\cite{gilyenea:svtrans}).

On the lower-bound side, $\Omega(1/\delta)$ for the cost of phase estimation has long been known to hold when the success probability is required to be a constant, this follows for instance from  the approximate counting lower bound of Nayak and Wu~\cite{NW99} (see also \cite{Bes05}).
Lin and Tong~\cite[Theorem~10]{linlin&tong:groundstateprep} proved lower bounds of $\Omega(1/\gamma)$ and $\Omega(1/\delta)$ on the cost for the setting with known eigenbasis and advice unitary (our row 6, and hence also row 8). This is subsumed by our stronger (and essentially optimal) $\Omega(1/\gamma\delta)$ lower bound in row~6.
As far as we are aware, ours is the first paper to systematically tie together these different results and to complete the table with tight upper and lower bounds for the cost in all 8 cases.

Let us also mention some recent work that is not directly covered by our results.
First, lower bounds for the slightly unusual small-success-probability regime were recently studied by Lin~\cite{Lin23}.
Second, there has been work to make phase estimation more efficient in the important special case where the unitary $U=e^{iH}$ is induced by  a Hamiltonian $H$ given classically as the sum of relatively simple terms, when the cost of phase estimation interacts with the cost of Hamiltonian simulation. See for instance the recent paper by Wan, Berta, and Campbell~\cite{WBC:randomizedphasest} and references therein.

\paragraph*{Application.}
She and Yuen~\cite[Theorems 1.6 and 1.7]{SY23} studied the \emph{$(t, \delta)$-Unitary recurrence time problem}, which is to distinguish whether an input unitary $U$ satisfies $U^t = I$ or $\| U^t - I \| \geq \delta$, promised that one of these is the case (see Definition~\ref{def: recurrencetime}). They proved non-matching upper and lower bounds for the cost of quantum algorithms for this problem (see Theorem~\ref{thm: SY} in this paper).
As an immediate application of our lower bound for fractional OR with advice, we also obtain improved lower bounds for the unitary recurrence time problem that match the upper bound of She and Yuen and answer one of their open problems~\cite[Section 2]{SY23}.

\begin{theorem}[Lower bound for Unitary recurrence time]\label{thm: recurrence time lower bound}
Any quantum algorithm solving the $(t,\delta)$-recurrence time problem for $N$-dimensional unitaries has cost $\Omega(t\sqrt{N}/\delta)$.
\end{theorem}

\noindent
Interestingly, our lower bound uses the adversary method as opposed to She and Yuen's usage of the polynomial method.

\subsection{Phase estimation with small error probability}

For our results in this subsection we revert to the original scenario of phase estimation, where an algorithm is given the actual eigenstate as input and the goal is to estimate its eigenphase~$\theta$. However, we now consider the regime where we want small error probability~$\eps$ rather than constant error probability~$1/3$. Let $\QPE_{N, \delta, \eps}$ denote the task of computing, with error probability $\leq\eps$, a $\delta$-approximation of~$\theta$.
Repeating Kitaev's $O(1/\delta)$-cost phase estimation algorithm~\cite{Kit95} $O(\log(1/\eps))$ times and taking the median of the answers, we have the following $\eps$-dependent upper bound.

\begin{theorem}[Kitaev + standard error-reduction]\label{thm: qpe upper bound}
For all integers $N \geq 2$, all $\eps \in (0,1/2)$, and all $\delta \in [0, 2\pi)$, there exists an algorithm that solves $\QPE_{N, \delta, \eps}$ with cost $O\rbra{\frac{1}{\delta}\log\frac{1}{\eps}}$.
\end{theorem}

Grover's algorithm~\cite{grover:search} can compute the $\OR_N$ function with error probability $\leq 1/3$ using  $O(\sqrt{N})$ queries to its $N$ input bits. 
Interestingly, there exists an $\eps$-error quantum algorithm for $\OR_N$  with only $O(\sqrt{N\log(1/\eps)})$ queries~\cite{BCWZ99}, which is asymptotically optimal. Similarly one can reduce error from $1/3$ to $\eps$ for all symmetric Boolean functions at the expense of only a factor $\sqrt{\log(1/\eps)}$ in the query complexity~\cite{wolf:degreesymmf}. This is a speed-up over the naive $O(\log(1/\eps))$ multiplicative overhead.
Since optimal quantum algorithms with error probability~$1/3$ for $\OR_N$ and for all symmetric functions  can be derived from phase estimation, one may ask if one can achieve such an efficient error-reduction for quantum phase estimation as well: is there an algorithm for $\QPE_{N, \delta, \eps}$ of cost $O\rbra{\frac{1}{\delta}\sqrt{\log(1/\eps)}}$? We answer this in the negative, showing Theorem~\ref{thm: qpe upper bound} is optimal.

\begin{theorem}\label{thm: qpe lb}
For integers $N \geq 2$ and $\eps \in (0,1/2), \delta \in (0,1)$, every algorithm that solves $\QPE_{N, \delta, \eps}$ has cost $\Omega\rbra{\frac{1}{\delta}\log\frac{1}{\eps}}$.
\end{theorem}

In particular, this means that the optimal complexity of $\OR_N$ with small error probability~$\eps$ of~\cite{BCWZ99} can\emph{not} be derived from a phase estimation routine, in contrast to the case of $\OR_N$ (and search) with constant error probability.
To show Theorem~\ref{thm: qpe lb} we first argue that a cost-$C$ algorithm for $\QPE_{N,\delta,\eps}$ gives us a cost-$C$ algorithm that distinguishes $U = I$ versus $U = I - (1 - e^{i\theta})\ketbra{0}{0}$ where $\theta \notin [-3\delta, 3\delta]$ mod $2\pi$. We then note that the acceptance probability of such an algorithm can be written as a degree-$2C$ \emph{trigonometric} polynomial in~$\theta$, and invoke a known upper bound on the growth of such trigonometric polynomials in order to lower bound their degree.

\section{Preliminaries}\label{sec: prelims}

We state the required preliminaries in this section.
All logarithms are taken base~2. For a positive integer $N$, $U(N)$ denotes the space of $N$-dimensional unitaries, and $I$ denotes the $N$-dimensional Identity matrix (we drop the subscript if the dimension is clear from context).

For a positive integer $N \geq 2$ and a value $\theta \in [0, 2\pi)$, define the $N$-dimensional unitary $U_\theta$ as $U_\theta = I - (1 - e^{i\theta})\ketbra{0}{0}$. In other words, $U_\theta$ is the diagonal matrix with all 1's except the first entry, which is $e^{i\theta}$.
For an integer $j \in \cbra{0, 1, \dots, N-1}$ and $\delta \in [0, 2\pi)$, define $M_{j, \delta} = I - (1 - e^{i\delta})\ketbra{j}{j}$.

\subsection{Model of computation}\label{sec: model}

Here we give a description of our model of computation for all tasks considered in this paper.
All problems considered in this paper have the following properties:
\begin{itemize}
    \item \textbf{Input:} An $N$-dimensional unitary $U$. We have access to the input as described below.
    \item \textbf{State space:} The state space of an algorithm comprises two registers: the first register is $N$-dimensional, and the second register is an arbitrarily large workspace.
    \item \textbf{Access to input and allowed operations:} An algorithm $\A$ may apply $U$ and $U^{-1}$ to the first register (possibly controlled by a qubit in the second register), and unitaries independent of $U$ to the whole space. It performs a POVM at the end to determine the classical output.
    \item \textbf{Cost of an algorithm:} Total number of applications of $U$ and $U^{-1}$ (or controlled versions thereof).
\end{itemize}
Depending on the specific problem under consideration, the following properties are variable.
\begin{itemize}
    \item \textbf{Initial state:} The initial state is assumed to be $\ket{0}\ket{0}$ unless mentioned otherwise.
    \item \textbf{Input promise:} The subset of the $N$-dimensional unitary group $U(N)$ from which the input is taken (possibly the full set).
    \item \textbf{Output:} The output requirement.
    \item \textbf{Advice:} We may be given access to a specific number of ``advice states'' $\ket{\alpha}$, or access to a specific number of applications of a unitary $A$ that prepares an advice state (e.g., $A\ket{0}=\ket{\alpha}$).
\end{itemize}

\subsection{Problems of interest}
We list our problems of interest here. All problems fit in the framework of the previous subsection, so we skip descriptions of the input, access to the input and allowed operations, and the workspace.

\begin{definition}[Phase Estimation]\label{def: qpe}
Let $N \geq 2$ be an integer and $\eps, \delta > 0$. The task $\QPE_{N, \delta, \eps}$ is:
\begin{itemize}
    \item \textbf{Advice:} We are given a single state $\ket{\psi}$ (in other words, our starting state is $\ket{\psi}\ket{0}$) with the promise that $U\ket{\psi} = e^{i\theta}\ket{\psi}$ for some unknown $\theta\in[0,2\pi)$.
    \item \textbf{Output:} With probability at least $1 - \eps$, output $\tilde{\theta} \in [0, 2\pi)$ such that $|\tilde{\theta} - \theta| \leq \delta~\textnormal{mod}~2\pi$.
\end{itemize}
\end{definition}

The following is a decision version of a special case of phase estimation, where $\ket{\psi}=\ket{0}$:

\begin{definition}
Let $N \geq 2$ be an integer, $\eps, \delta \in (0, 1)$.
The task $\dist_{N, \delta, \eps}$ is:
\begin{itemize}
    \item \textbf{Input promise:} $U \in \cbra{I, \cbra{U_\theta: \theta \notin [-\delta, \delta] \mod 2\pi}}$.
    \item \textbf{Output:} With probability at least $1 - \eps$, output 1 if $U = I$ and output 0 if $U\neq I$.
\end{itemize}
\end{definition}

We next define the natural variant of phase estimation that we consider when an algorithm need not be given a state from the target eigenspace but only a state $\ket{\alpha}$ that has non-negligible overlap with that eigenspace.

\begin{definition}[Maximum phase estimation]\label{def: maxphaseestimation}
Let $N \geq 2$ be an integer and $\delta > 0$. The task $\maxqpe_{N, \delta}$ is:
\begin{itemize}
    \item \textbf{Input promise:} We consider two cases: one where the eigenbasis of $U$ is known, and the other where it is unknown. In the former case, we may assume $U = \sum_{j = 0}^{N - 1}e^{i\theta_j}\ketbra{j}{j}$ for $\theta_j\in[0,2\pi)$. Define $\theta_{\max}=\max_{j \in \cbra{0, 1, \dots, N-1}}\theta_j$. We are promised that $\theta_{\max}\in[0,2\pi-2\delta)$.
    \item \textbf{Advice:} We consider two cases: 
    \begin{itemize}
        \item In one case we are given access to advice in the form of a state $\ket{\alpha}$ such that $\|P_S\ket{\alpha}\|^2 \geq \gamma^2$, where $P_S$ denotes the projection on $S$, the space of all eigenstates with eigenphase~$\theta_{\max}$. If $S$ is spanned by one $\ket{u_{\max}}$, this requirement is the same as $|\braket{\alpha}{u_{\max}}| \geq \gamma$. 
        \item In the other case, we have black-box access to a unitary~$A$ that prepares such a state~$\ket{\alpha}$. We can apply $A$ and $A^{-1}$. As before, $\gamma$ is the \emph{overlap} of $\ket{\alpha}$ with the target eigenspace.
    \end{itemize}
    \item \textbf{Number of accesses to advice:} We either have `few' accesses to advice as defined above ($o(1/\gamma^2)$ advice states or $o(1/\gamma)$ advice unitaries), or `many' accesses to advice ($\Omega(1/\gamma^2)$ advice states or $\Omega(1/\gamma)$ advice unitaries).
    \item \textbf{Output:} With probability at least $2/3$, output a value in $[\theta_{\max}-\delta,\theta_{\max}+\delta]~\textnormal{mod}~2\pi$.
\end{itemize}
\end{definition}

\begin{definition}[Unitary recurrence time, {\cite[Definition 1.5]{SY23}}]\label{def: recurrencetime}
    For integers $N \geq 2, t \geq 1$ and $\delta \in (0,1)$, the $(t,\delta)$-recurrence time problem is:
    \begin{itemize}
        \item \textbf{Input promise:} Either $U = I$, or $\|U^t - I\| \geq \delta$ in spectral norm.
        \item \textbf{Output:} With probability at least $2/3$: output 1 if $U = I$, and 0 otherwise.
    \end{itemize}
\end{definition}

\noindent
The following are the non-matching upper and lower bounds for this problem of She and Yuen~\cite{SY23} (which we improve upon in our  Theorem~\ref{thm: recurrence time lower bound}).

\begin{theorem}[{\cite[Theorems 1.6 and 1.7]{SY23}}]\label{thm: SY}
    Let $\delta \leq \frac{1}{2\pi}$. Every quantum algorithm solving the $(t,\delta)$-recurrence time problem for $d$-dimensional unitaries has cost $\Omega\rbra{\max\rbra{t/\delta, \sqrt{d}}}$. The $(t,\delta)$-recurrence time problem can be solved with cost $O(t\sqrt{d}/\delta)$.
\end{theorem}

\subsection{Trigonometric polynomials and their growth}

\begin{definition}[Trigonometric Polynomials]
A function $p : \R \to \C$ is said to be a \emph{trigonometric polynomial} of degree $d$ if there exist complex numbers $\cbra{a_k : k \in \cbra{-d, \dots, d}}$ such that for all $\theta \in \R$,
\[
p(\theta) = \sum_{k = -d}^da_ke^{ik\theta}.
\]
\end{definition}

We will use the following property of low-degree trigonometric polynomials.

\begin{theorem}[{\cite[Theorem~5.1.2]{BE95}}]\label{thm: trig magic}
Let $t$ be a degree-$n$ real-valued trigonometric polynomial and $s \in (0, \pi/2]$ be such that $\mu(\cbra{\theta \in [-\pi, \pi) : |t(\theta)| \leq 1}) \geq 2\pi - s$, where $\mu$ denotes the Lebesgue measure on $\mathbb{R}$. Then,
$\sup_{x \in \mathbb{R}}|t(x)| \leq \exp(4ns)$.
\end{theorem}

\section{Lower bounds for maximum phase estimation and Unitary recurrence time}\label{sec: lower}

In this section we show lower bounds on the quantum complexity of maximum phase estimation obtained by varying all its parameters (see Section~\ref{sec: model} and Definition~\ref{def: maxphaseestimation}). In this section and the next, we refer to the row numbers of Table~\ref{table: bounds} when stating and proving our bounds.

Recall that for an integer $j \in \cbra{0, 1, \dots, N-1}$ and $\delta \in [0, 2\pi)$ we define $M_{j, \delta} = I - (1 - e^{i\delta})\ketbra{j}{j}$.
Our lower bounds will be via reduction from the following ``Fractional OR with advice'' problem, which fits in the framework of the model described in Section~\ref{sec: model}.

\begin{definition}[Fractional OR with advice]\label{def: gndeltat}
    Let $N \geq 2$ be   integer, $\delta > 0$. The task $\for_{N, \delta, t}$ is:
    \begin{itemize}
        \item \textbf{Input promise:} $U \in \cbra{I, \cbra{M_{j, \delta} : j \in \cbra{1, 2, \dots, N-1}}}$.
        \item \textbf{Advice:} When $U = I$ we are given $t$ copies of $\ket{0}$ as advice. When $U = M_{j, \delta}$, we are given $t$ copies of the state $\gamma\ket{j} + \sqrt{1 - \gamma^2}\ket{0}$, i.e., part of our starting state is $(\gamma\ket{j} + \sqrt{1 - \gamma^2}\ket{0})^{\otimes t}$.
        \item \textbf{Output:} With probability at least $2/3$, output 1 if $U = I$ and output 0 if $U \neq I$.
    \end{itemize}
\end{definition}

We first show a lower bound on the cost of computing $\for_{N,\delta,t}$ when $t = o(1/\gamma^2)$.
All of our lower bounds in Table~\ref{table: bounds} as well as our lower bound for the Unitary recurrence time problem will use this lower bound. The proof (given in Appendix~\ref{app: adversary proof}) follows along the same lines as Ambainis' adversary lower bound~\cite[Theorem~4.1]{Amb02} of $\Omega(\sqrt{N})$ queries for the $N$-bit Search problem, but now we additionally take into account the initial advice states and the fact that our input unitaries are only \emph{fractional} versions of phase queries.

\begin{theorem}\label{thm: gndeltat lower bound}
     For an integer $N\geq 2$, real numbers $\gamma \geq 1/\sqrt{N}$, $\delta \in (0,1)$ and $t = o(1/\gamma^2)$, every  algorithm solving $\for_{N, \delta, t}$ has cost $\Omega(\sqrt{N}/\delta)$.
\end{theorem}

In the following four lower-bound lemmas we assume $\delta\in (0,1)$ as well.

\begin{lemma}[Lower bound for Rows 1,3]\label{lemma:13lb}
Row 1 (and hence Row 3) has a lower bound of $\Omega(\sqrt{N}/\delta)$.
\end{lemma}

\begin{proof}
    A cost-$C$ algorithm $\mathcal{A}$ for $\maxqpe_{N, \delta}$ with $t$ advice states and known eigenbasis of~$U$ immediately yields a cost-$C$ algorithm $\A'$ for $\for_{N, 3\delta, t}$: run $\mathcal{A}$ on the input unitary, output 1 if the output phase is in $[- \delta, \delta]$ modulo $2\pi$, and output 0 otherwise.
    When $U = I$, the correctness of $\A$ guarantees that with probability at least $2/3$, the value output by $\A$ is in $[-\delta, \delta]$ mod $2\pi$. When $U = M_{j, 3\delta}$, the correctness of $\A$ guarantees that with probability at least $2/3$, the value output by $\A$ is in $[2\delta, 4\delta]$. For $\delta \in (0,1)$, we have $[-\delta, \delta]\mod 2\pi \cap [2\delta, 4\delta]$ mod $2\pi = \emptyset$. Thus, $\A'$ solves $\for_{N, 3\delta, t}$ and has cost $C$.
    Theorem~\ref{thm: gndeltat lower bound} yields the bound  $C=\Omega\rbra{\sqrt{N}/\delta}$ when $t = o(1/\gamma^2)$, giving the desired result.
\end{proof}

\begin{lemma}[Lower bound for Rows 2,4]\label{lemma:24lb}
Row 2 (and hence Row 4) has a lower bound of $\Omega\rbra{1/\gamma\delta}$.
\end{lemma}

\begin{proof}
    We prove the required lower bound for $\maxqpe_{N, \delta}$ with inputs satisfying the promise that $U \in \cbra{I_N, \cbra{M_{j, 3\delta} : j \in \cbra{1, 2, \dots, 1/\gamma^2 - 1}}}$. Because of this assumption, we may take the uniform superposition over the first $1/\gamma^2$ computational basis states as our advice state: the algorithm should work with such an advice state, since it has overlap at least~$\gamma$ with the top eigenspace for each of the possible~$U$. However, an algorithm can prepare such advice states at no cost, so we may assume that the algorithm has no access to advice at all. As in the previous proof, this gives an algorithm of the same cost for $\for_{1/\gamma^2, 3\delta, 0}$ (ignoring all other dimensions). Theorem~\ref{thm: gndeltat lower bound} with $N = 1/\gamma^2$  and $t = 0$ yields the required lower bound of $\Omega(1/\gamma\delta)$.
\end{proof}

\begin{lemma}[Lower bound for Rows 5,7]\label{lemma:57lb}
Row 5 (and hence Row 7) has a lower bound of $\Omega(\sqrt{N}/\delta)$.
\end{lemma}

\begin{proof}
    Towards the required lower bound, consider a cost-$C$ algorithm $\mathcal{A}$ solving $\maxqpe_{N, \delta}$ with inputs satisfying the promise $U \in \cbra{I_N, \cbra{M_{j, 3\delta} : j \in \cbra{1, 2, \dots, N-1}}}$, and with $t = o(1/\gamma)$ accesses to a unitary that prepares an advice state that has overlap at least $\gamma$ with the target eigenspace. We want to construct an algorithm $\A'$ for $\maxqpe_{N, \delta}$ with the same promised inputs that uses \emph{no} advice, and with cost not much larger than that of~$\A$. 
    Note that we may assume $\gamma = o(1)$, since otherwise $t = 0$, so then $\A$ itself already uses no advice.

 We first show how an algorithm can itself implement a good-enough advice unitary~$A$ quite cheaply. 
 Assuming without loss of generality that $k=\pi/(3\delta)$ is an integer, $U^k$ is actually a ``phase query'': if $U = M_{j, 3\delta}$, then $U^k$ is the diagonal matrix with  $(e^{i 3\delta})^k=e^{i\pi}=-1$ in the $j$th entry and 1s elsewhere; and if $U=I$ then $U^k=I$. Thus $A$ can start by mapping $\ket{0}$ to a uniform superposition over all indices, and then use Grover's algorithm with $U^k$ as the phase-query operator to amplify the amplitude of $\ket{j}$ to $\geq\gamma$.
 We know that $O(\gamma\sqrt{N})$ ``Grover iterations'' suffice for this (see, for example,~\cite[Section~7.2]{lecturenotes} for details).
 Each Grover iteration would use one phase-query $U^k$, so the overall cost (number of applications of $U$ and $U^{-1}$) of this advice unitary is $k\cdot O(\gamma\sqrt{N})=O(\gamma\sqrt{N}/\delta)$.
 If $U=I$, the state just remains the uniform superposition that Grover's algorithm starts with.
 
 We now have all components to describe $\A'$: Run $\A$, and whenever $\A$ invokes an advice unitary, use the above~$A$. Since $\A$ uses at most $t$ advice unitaries, the cost of $\A'$ is at most $C + t\cdot O(\gamma\sqrt{N}/\delta)$. Note that $\A'$ uses no advice at all anymore, and solves $\maxqpe_{N,\delta}$ under the promise that the input unitary satisfies $U \in \cbra{I_N, \cbra{M_{j, 3\delta} : j \in \cbra{1, 2, \dots, N-1}}}$. Again, this immediately yields an algorithm of the same cost for $\for_{N,3\delta,0}$ as in the previous two proofs. Theorem~\ref{thm: gndeltat lower bound} now implies
    $$
    C + O(t\gamma\sqrt{N}/\delta) = \Omega(\sqrt{N}/\delta),
    $$
    and hence $C=\Omega(\sqrt{N}/\delta)$ since $t = o(1/\gamma)$ ($t\leq c/\gamma$ for sufficiently small constant $c$ also suffices).
\end{proof}

\begin{lemma}[Lower bound for Rows 6,8]\label{lemma:68lb}
Row 6 (and hence Row 8) has a lower bound of $\Omega\rbra{1/\gamma\delta}$.
\end{lemma}

\begin{proof}
    Just as in the proof of Lemma~\ref{lemma:24lb}, we may assume $N = 1/\gamma^2$ by only allowing input unitaries of the form $U \in \cbra{I_N, \cbra{M_{j, 3\delta} : j \in \cbra{1, 2, \dots, 1/\gamma^2-1}}}$. With this assumption, we may assume that we have no access to advice (i.e., $t = 0$) since an algorithm can prepare a good-enough advice state at no cost, namely the uniform superposition over all $N=1/\gamma^2$ basis states. Lemma~\ref{lemma:57lb} now yields the required lower bound of $\Omega(1/\gamma\delta)$.
\end{proof}

Finally we prove an optimal lower bound for the Unitary recurrence time problem, matching She and Yuen's upper bound (Theorem~\ref{thm: SY}), resolving one of their open problems~\cite[Section~2]{SY23}.

\begin{proof}[Proof of Theorem~\ref{thm: recurrence time lower bound}]
Consider an algorithm $\mathcal{A}$ solving the $(t,\delta)$-recurrence time problem. Restrict to inputs of the form $U \in \cbra{I_N, \cbra{M_{j, 3\delta/t} : j \in \cbra{1, 2, \dots, N-1}}}$. 
When $U = I$ we have $U^t = I$. When $U = M_{j, 3\delta/t}$, we have $\| U^t - I\| = |1 - e^{3i\delta}| \geq \delta$ for all $\delta \in [0,1]$. Thus, $\mathcal{A}$ solves $\for_{N, 3\delta/t, 0}$. Theorem~\ref{thm: gndeltat lower bound} yields the required lower bound of $\Omega(t\sqrt{N}/\delta)$.
\end{proof}

\section{Upper bounds for maximum phase estimation}\label{sec: upper}

In this section we show upper bounds on the quantum complexity of our 8 variants of maximum phase estimation (see Section~\ref{sec: model}, Definition~\ref{def: maxphaseestimation} and Table~\ref{table: bounds}).
We require the following generalized maximum-finding procedure, adapted from \cite[Theorem~49]{AGGW20}; we changed their wording a bit and modified it from minimum-finding to maximum-finding.

\begin{lemma}[Generalized maximum-finding {{\cite[Theorem~49]{AGGW20}}}]\label{lemma: genmaxfind}
There exists a quantum algorithm $\cal M$ and constant $C>0$ such that the following holds.
Suppose we have a $q$-qubit unitary $V$ such that 
\[
V\ket{0}=\sum_{k=0}^{K-1}\ket{\psi_k}\ket{x_k},
\]
where $x_0>x_1>\cdots> x_{K-1}$ are distinct real numbers (written down in finite precision), and the $\ket{\psi_k}$ are unnormalized states.
Let $X$ be the random variable obtained if we were to measure the last register, so $\Pr[X=x_k]=\norm{\ket{\psi_k}}^2$.
Let $M\geq C/\sqrt{\Pr[X \geq x_j]}$ for some $j$. 
Then $\cal M$ uses at most $M$ applications of $V$ and $V^{-1}$, and $O(qM)$ other gates, 
and outputs a state $\ket{\psi_i}\ket{x_i}$ (normalized) such that $x_i\geq x_j$ with probability at least $3/4$ (in particular, if $j=0$ then $\cal M$ outputs the maximum).
\end{lemma}

\begin{remark}\label{rmk:v_usage_reflections}
    It may be verified by going through \cite[Lemma~48 \& Theorem~49]{AGGW20} that the only applications of $V$ and $V^{-1}$ used by $\cal M$ are to prepare $V\ket{0}$ starting from $\ket{0}$, and to reflect about the state $V\ket{0}$.
\end{remark}

\noindent
We can use generalized maximum-finding to approximate the largest eigenphase starting from the ability to prepare a superposition of eigenstates (possibly with some additional workspace qubits): 

\begin{lemma}\label{lemma: maxthetafind}
There exists a quantum algorithm $\cal B$ such that the following holds.
Suppose we have an $N$-dimensional unitary $U$ with (unknown) eigenstates $\ket{u_0},\ldots,\ket{u_{N-1}}$ and associated eigenphases $\theta_0,\ldots,\theta_{N-1}\in[0,2\pi-2\delta)$.
Suppose we also have a unitary $A$ such that
\[
A\ket{0}=\sum_{j=0}^{N-1}\alpha_j\ket{u_j}\ket{\phi_j},
\]
where $\sum_{j:\theta_j=\theta_{\max}}|\alpha_j|^2\geq\gamma^2$ and the $\ket{\phi_j}$ are arbitrary (normalized) states.
Then $\cal B$ uses at most $O(1/\gamma)$ applications of $A$ and $A^{-1}$, and $O(\log(1/\gamma)/\gamma\delta)$ applications of $U$ and $U^{-1}$, 
and with probability at least $2/3$ outputs a number $\theta\in[\theta_{\max}-\delta,\theta_{\max}+\delta]$.
\end{lemma}

\begin{proof}
As mentioned in footnote~\ref{note:hannote}, by multiplying $U$ with the known phase $e^{i\delta}$, we may assume all eigenphases $\theta_j$ are in the interval $[\delta,2\pi-\delta)$.
Let $\tilde{V}$ be the unitary that applies phase estimation with unitary~$U$, precision $\delta$, and small error probability $\eta$ (to be determined later), on the first register of the state $A\ket{0}$, writing the estimates of the phase in a third register.
Then
\[
\tilde{V}\ket{0}=\sum_{j=0}^{N-1}\alpha_j\ket{u_j}\ket{\phi_j}\ket{\tilde{\theta_j}},
\]
where, for each $j$, $\ket{\tilde{\theta_j}}$ is a superposition over estimates of $\theta_j$, 
most of which are $\delta$-close to $\theta_j$.

For the purposes of analysis, we would like to define a ``cleaned up'' unitary~$V$ (very close to~$\tilde{V}$) that doesn't have any estimates with error $>\delta$ in $V\ket{0}$
Let $\ket{\tilde{\theta_j}'}$ be the state obtained from $\ket{\tilde{\theta_j}}$ by removing the estimates that are more than $\delta$-far from $\theta_j$, and renormalizing.
Because we ran phase estimation with error probability $\leq\eta$, it is easy to show that 
$\norm{\ket{\tilde{\theta_j}'}-\ket{\tilde{\theta_j}}}=O(\sqrt{\eta})$.
Then there exists\footnote{This is fairly easy to show, see e.g.~\cite[proof of Theorem 2.4 in Appendix~A]{chen&wolf:linreg}.} a unitary $V$ such that $\norm{\tilde{V}-V}=O(\sqrt{\eta})$
and
\[
V\ket{0}=\sum_{j=0}^{N-1}\alpha_j\ket{u_j}\ket{\phi_j}\ket{\tilde{\theta_j}'}=\sum_{k=0}^{K-1}\ket{\psi_k}\ket{x_k},
\]
where the $x_k$ are the distinct estimates that have support in the last register, and the $\ket{\psi_k}$ are (unnormalized) superpositions of the $\ket{u_j}\ket{\phi_j}$'s that are associated with those estimates.

Algorithm $\cal B$ now applies the maximum-finding algorithm~$\cal M$ of Lemma~\ref{lemma: genmaxfind} with the unitary~$\tilde{V}$. Let us first analyze what would happen if $\cal B$ used the cleaned-up $V$ instead of~$\tilde{V}$.
Because we assumed that $\theta_j\in[\delta,2\pi-\delta)$ for all $j$, and we have removed all estimates of $\theta_j$ with error $>\delta$, the largest $x_k$'s in $V\ket{0}$ are good estimates of $\theta_{\max}$ (with no worries about potential ``wrapping around'' $2\pi$ of the estimates).
Let $X$ denote the random variable obtained if we measure the last register of~$V\ket{0}$, and note that $\Pr[X\geq \theta_{\max}-\delta]\geq\sum_{j:\theta_j=\theta_{\max}}|\alpha_j|^2\geq\gamma^2$ because all estimates in $V\ket{0}$ have error $\leq\delta$. Hence $\cal B$ would use $O(1/\gamma)$ applications of $V$ and $V^{-1}$ to find a $\theta\in[\theta_{\max}-\delta,\theta_{\max}+\delta]$ with success probability $\geq 3/4$.

Algorithm $\cal B$ will actually use $\tilde{V}$ and $\tilde{V}^{-1}$ instead of $V$ and $V^{-1}$, which (because errors in quantum circuits add at most linearly) incurs an overall error in operator norm of $\leq O(\sqrt{\eta})\cdot O(1/\gamma)$. Choosing $\eta\ll\gamma^2$, this overall error can be made an arbitrarily small constant. The success probability of the algorithm can drop slightly below $3/4$ now due to this error, but is still~$\geq 2/3$.

It remains to analyze the cost of $\cal B$.
Each $\tilde{V}$ uses 1 application of $A$, and $O(\log(1/\eta)/\delta)=O(\log(1/\gamma)/\delta)$ applications of $U$ and $U^{-1}$ for phase estimation (Theorem~\ref{thm: qpe upper bound}), so $\cal B$ uses $O(1/\gamma)$ applications of $A$ and $A^{-1}$, and $O(\log(1/\gamma)/\gamma\delta)$ applications of $U$ and $U^{-1}$ in total.
\end{proof}

The upper bounds for our 8 variants of phase estimation (see Table~\ref{table: bounds}) will all follow from this.
To derive these upper bounds, we start with the 4 odd-numbered rows, where it turns out the advice is not actually needed to meet our earlier lower bounds.
The next proof is basically the same as \cite[Lemma~50]{AGGW20} about estimating the minimal eigenvalue of a Hamiltonian (this improved slightly upon the earlier result of~\cite{poulin2009GibbsSamplingAndEval}; see also~\cite[Lemma~3.A.4]{gilyen:thesis}).

\begin{lemma}[Upper bound for Rows 1, 3, 5, 7]\label{lemma: odd ub}
There is an algorithm that uses no advice and solves the case in Row 3 (and hence in Rows 1, 5, and 7 as well) with cost $O(\sqrt{N}\log(N)/\delta)$.
\end{lemma}

\begin{proof}
Let $A$ be the unitary that maps $\ket{0}$ to the maximally entangled state in $N$ dimensions. This state can be written in any orthonormal basis, including the (unknown) eigenbasis of $U$:
\[
A\ket{0}=\frac{1}{\sqrt{N}}\sum_{j=0}^{N-1}\ket{j}\ket{j}=
\frac{1}{\sqrt{N}}\sum_{j=0}^{N-1}\ket{u_j}\ket{\overline{u_j}},
\]
where $\ket{\overline{u_j}}$ denotes the entry-wise conjugated version of $\ket{u_j}$.
Applying Lemma~\ref{lemma: maxthetafind} with this $A$, $\ket{\phi_j}=\ket{\overline{u_j}}$, and $\gamma=1/\sqrt{N}$ gives the result.
\end{proof}

The next two lemmas cover the 4 upper-bound cases where advice states/unitaries \emph{are} helpful.

\begin{lemma}[Upper bound for Rows~6, 8]\label{lemma: 68 ub}
There is a quantum algorithm that uses $O(1/\gamma)$ applications of the advice unitary (and its inverse) and solves the case in Row~8 (and hence the case in Row~6 as well) with cost $O(\log(1/\gamma)/\gamma\delta)$.
\end{lemma}

\begin{proof}
Apply Lemma~\ref{lemma: maxthetafind} with the unitary $A$ that maps $\ket{0}$ to $\ket{\alpha}$, with empty states $\ket{\phi_j}$.
\end{proof}

\begin{lemma}[Upper bound for Rows 2, 4]\label{lemma: 24 ub}
There is a quantum algorithm that uses $O(1/\gamma^2)$ copies of the advice state and solves the case in Row 4 (and in Row 2) with cost $O(\log(1/\gamma)/\gamma\delta)$.
\end{lemma}

\begin{proof}
We will build upon the algorithm for Row~8 of Lemma~\ref{lemma: 68 ub}. By Remark~\ref{rmk:v_usage_reflections} and the algorithm in Lemma~\ref{lemma: 68 ub}, its $O(1/\gamma)$ applications of the advice unitary~$A$ and its inverse~$A^{-1}$ are only used there for two purposes: (1) to prepare a copy of the advice state~$A\ket{0}=\ket{\alpha}$, and (2) to reflect about~$\ket{\alpha}$. We now want to replace these applications of $A$ by using copies of the advice state. For (1) this is obvious.
 Assume the algorithm for Row~8 uses (2) at most $C/\gamma$ times, for some constant~$C$.
To implement these reflections, we will invoke the result of Lloyd, Mohseni, and Rebentrost~\cite{lmr:pca} (see also~\cite{kimmelea:hamsim}), who showed that given a number $t>0$ and $O(t^2/\eta)$ copies of a mixed quantum state~$\rho$, one can implement the unitary $e^{i t \rho}$
up to error~$\eta$ (in diamond-norm difference between the intended unitary and the actually-implemented channel).
We will use this result with $\rho=\ketbra{\alpha}{\alpha}$, $t=\pi$, $\eta=\gamma/(100C)$, noting that the implemented unitary
$
e^{i\pi \ketbra{\alpha}{\alpha}}=I-2\ketbra{\alpha}{\alpha}
$
is a reflection about $\ket{\alpha}$ (up to a global minus sign that doesn't matter).

Accordingly, we can implement the $\leq C/\gamma$ reflections used by the algorithm for Row~8 using $O(1/\gamma^2)$ copies of $\ket{\alpha}$, each reflection implemented with error~$\leq\eta$.
Because errors in quantum circuits add at most linearly, the overall error between the algorithm of Row~8 and our simulation of it (using copies of $\ket{\alpha}$) is at most $\eta \cdot C/\gamma\leq 1/100$.
Hence we obtain an algorithm for Row~4 that uses $O(1/\gamma^2)$ copies of~$\ket{\alpha}$ and has the same cost $O(\log(1/\gamma)/\gamma\delta)$ as the algorithm of Row~8.
\end{proof}

\section{Tight bounds for phase estimation with small error probability}

Here we prove our lower bound for quantum algorithms solving phase estimation with precision $\delta$ and error probability at most $\eps$, Theorem~\ref{thm: qpe lb}, which follows from Claims~\ref{claim: reduction to decision problem} and~\ref{claim: decision problem lower bound} below.

\begin{claim}\label{claim: reduction to decision problem}
For all integers $N\geq 2$, all $\eps \in (0, 1/2)$ and $\delta \in (0, 1)$, if there is a cost-$d$ algorithm solving $\QPE_{N, \delta, \eps}$, then there is a cost-$d$ algorithm solving $\dist_{N, \delta,\eps}$.
\end{claim}

\begin{proof}
Consider an algorithm $\mathcal{A}$ of cost $d$ that solves $\QPE_{N, \delta, \eps}$. 
We construct below an algorithm $\A'$ of cost $d$ solving $\dist_{N, \delta, \eps}$.
Let $U \in U(N)$ be the input. The following is the description of $\A'$:
\begin{enumerate}
    \item Run $\mathcal{A}$ with inputs $U$ and $\ket{0}$.
    \item Output 1 if the output of $\A$ is in $[-\delta, \delta] \mod 2\pi$, and output 0 otherwise.
\end{enumerate}
Clearly $\A'$ is a valid algorithm, as far as access to input and allowed operations are concerned, since its initial state is $\ket{0}$, it applies $U, U^{-1}$, and some unitaries independent of $U$, and finally performs a two-outcome projective measurement to determine the output bit. The cost of $\A'$ is $d$.

The correctness follows along the same lines as the proofs in Section~\ref{sec: lower}. We prove correctness here for completeness. First note that the state $\ket{0}$ is an eigenstate of all $U \in \cbra{I} \cup \cbra{U_\theta : \theta \notin [-3\delta, 3\delta] \mod 2\pi}$.
When $U = I$, the correctness of $\A$ guarantees that with probability at least $1 - \eps$, the value output by $\A$ is in $[-\delta, \delta]\mod 2\pi$. When $U = U_\theta$, the correctness of $\A$ guarantees that with probability at least $1 - \eps$, the value output by $\A$ is in $[\theta-\delta, \theta + \delta]\mod 2\pi$.
For $\theta \notin [-3\delta, 3\delta] \mod 2\pi$ we have $[-\delta, \delta]\mod 2\pi \cap [\theta-\delta, \theta + \delta]\mod 2\pi = \emptyset$ since $\delta < 1 < 2\pi/5$, and hence $\mathcal{A}'$ solves $\dist_{N, \delta, \eps}$.
\end{proof}

We next show a lower bound for the cost of algorithms computing $\dist_{N, \delta, \eps}$.

\begin{claim}\label{claim: decision problem lower bound}
For all integers $N \geq 2$, all $\eps \in (0, 1/2)$ and $\delta \in (0, 1)$, every algorithm for $\dist_{N, \delta, \eps}$ has cost $\Omega\rbra{\frac{1}{\delta}\log\frac{1}{\eps}}$.
\end{claim}

In order to prove Claim~\ref{claim: decision problem lower bound}, we first show that amplitudes of basis states in low-cost algorithms that run on $U_\theta$ are low-degree trigonometric polynomials in $\theta$. This is analogous to the fact that amplitudes of basis states in query algorithms for Boolean functions are low-degree (algebraic) polynomials in the input variables~\cite[Lemma~4.1]{BBCMW01}, and our proof is inspired by theirs.

\begin{claim}\label{claim: acc prob trig poly}
Let $t > 0$ be a positive integer and let $\theta \in [0, 2\pi]$. Consider a quantum circuit that has starting state $\ket{0}$, uses an arbitrary number of $\theta$-independent unitaries, uses $t$ applications of controlled-$U_\theta$ and controlled-$U_\theta^{-1}$ in total, and performs no intermediate measurements. Then the amplitudes of basis states before the final measurement are degree-$t$ trigonometric polynomials in $\theta$.
\end{claim}

\begin{proof}
We prove this by induction on $t$.
The claim is clearly true when $t = 0$ since all amplitudes are constants in this case.
For the inductive step, suppose the claim is true for $t = d$. Let $\ket{\psi_d}$ denote the state of the circuit just before the application of the $(d+1)$th application of $U_\theta$ (the argument for $U_\theta^{-1}$ is similar, and we skip it).
By the inductive hypothesis, we have
\[
\ket{\psi_d} = \sum_w\sum_{b \in \zone}\sum_{j = 0}^{N-1}p_{j,b,w}(\theta)\ket{j}\ket{b}\ket{w},
\]
where the first register is where $U_\theta$ and $U_\theta^{-1}$ act, the second register is the control qubit, and the last register represents the workspace (i.e., $U_\theta$ and $U_\theta^{-1}$ do not act on this register), and each $p_{j, b, w}$ is a trigonometric polynomial of degree at most $d$ in $\theta$.
For a basis state $\ket{j}\ket{b}\ket{w}$, we have
\[
U_\theta \ket{j}\ket{b}\ket{w} = \begin{cases}
e^{i\theta}\ket{0}\ket{b}\ket{w} & \mbox{if }j = 0 \mbox{ and }b = 1\\
\hspace*{1.2em}\ket{j}\ket{b}\ket{w} & \text{otherwise}.
\end{cases}
\]
In both cases, the amplitudes of the basis states after the application of $U_\theta$ are degree-$(d + 1)$ trigonometric polynomials in $\theta$. 
After the last application of $U_\theta$ the algorithm will apply an input-independent unitary. The amplitudes after that unitary are linear combinations of the amplitudes before, which won't increase degree. 
This concludes the inductive step, and hence the theorem.
\end{proof}

\begin{proof}[Proof of Claim~\ref{claim: decision problem lower bound}]

Consider a cost-$t$ algorithm $\A'$ solving $\dist_{N, \delta, \eps}$.
Claim~\ref{claim: acc prob trig poly} implies that on input $U_\theta$, the amplitudes of the basis states before the final  measurement are degree-$t$ trigonometric polynomials in $\theta$. The acceptance-probability polynomial $p : \R \to \R$ given by
$
p(\theta) : = \Pr[\mathcal{A}'(U_\theta) = 1]
$
is a degree-$2t$ trigonometric polynomial, because it is the sum of squares of moduli of certain amplitudes, and each of these squares is a degree-$2t$ trigonometric polynomial.
The correctness of the algorithm ensures that $p(0)\in[1 - \eps,1]$ and $p(\theta) \in [0, \eps]$ for all $\theta \notin [-3\delta, 3\delta] \mod 2\pi$. See Figure~\ref{fig: successprobfig} for a visual depiction of the behaviour of $p$ for $\theta \in [-\pi, \pi)$.

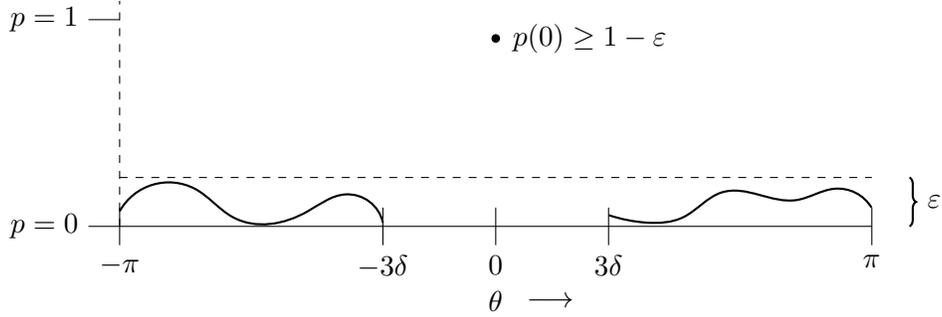
\begin{figure}[hbt]
\begin{center}
    \begin{tikzpicture}[scale=0.5]
    \draw (-10,0)-- (10,0); 
    \node (xaxislabel) at (0, -2) {$\theta$};
    \node (p0) at (-12,0) {$p = 0$};
    \node (p1) at (-12,5.5) {$p = 1$};
    \draw[dashed] (-10,0) -- (-10,6);
    \draw (p0) -- (-10,0);
    \draw (p1) -- (-10,5.5);
    \node[right=.5mm of xaxislabel] (rightarrow) {$\longrightarrow$};
    \draw (0,.5) -- (0,-.5) node[below] {$0$};
    \draw (-10,.5) -- (-10,-.5) node[below] {$-\pi$};
    \draw (10,.5) -- (10,-.5) node[below] {$\pi$};
    \draw (-3,.5) -- (-3,-.5) node[below] {$-3\delta$};
    \draw (3,.5) -- (3,-.5) node[below] {$3\delta$};
    \draw[fill=black] (0,5) circle (0.1);
    \draw [thick] (-10 ,.4) to [ curve through ={(-8,1)  . . (-7,.25) . . (-5,.4)  . . (-4,.85)}] (-3,.1);
    \draw [thick] (3 ,.3) to [ curve through ={(4,.1)  . . (5,.25) . . (6,.9)  . . (8,.7) . . (9,1)}] (9.99,.5);
    \draw[thick, decorate,decoration={brace}] (11,1.3) --
    node[right=1mm]{$\eps$} (11,0);
    \draw [dashed] (-10,1.3) -- (10,1.3);
    \node (0label) at (2.5,5) {$p(0) \geq 1 - \eps$};
\end{tikzpicture}
\end{center}
\caption{Acceptance probability $p$ of $\A'$ as a function of $\theta$ in the proof of Claim~\ref{claim: decision problem lower bound}}
\label{fig: successprobfig}
\end{figure}

Scaling by a global factor of $1/\eps$, we obtain a trigonometric polynomial $q$ of degree $2t$ satisfying:
\begin{itemize}
    \item $q(0) \geq (1 - \eps)/\eps > 1/(2\eps)$, and
    \item $q(\theta) \in [0, 1]$ for all $\theta \in [-\pi, \pi) \setminus [-3\delta, 3\delta]$.
\end{itemize}
Thus, Theorem~\ref{thm: trig magic} is applicable with $s = 6\delta$ and $n = 2t$, which implies
$
1/(2\eps) \leq \sup_{x \in \mathbb{R}}|q(x)| \leq \exp(48t\delta).
$
By taking logarithms and rearranging we get $t = \Omega\rbra{\frac{1}{\delta} \log\frac{1}{\eps}}$,
proving the theorem.
\end{proof}

\section{Conclusion}

In this paper we considered several natural variants of the fundamental phase estimation problem in quantum computing, and proved essentially tight bounds on their cost in each setting. As an immediate application of one of our bounds, we resolved an open question of~\cite[Section~2]{SY23}.

We mention some interesting questions in the first variant of phase estimation we considered, where an algorithm is given a number of copies of advice states/unitaries instead of black-box access to a perfect eigenstate as in the basic phase estimation setup. First, are the logarithmic overheads for the cost in the input dimension~$N$ and the inverse of the overlap~$\gamma$ in our upper bounds (see Table~\ref{table: bounds}) necessary, or can we give tighter upper bounds? 
Second, what is the optimal gate-complexity for rows~2 and~4?
Third, can we show the $\log(1/\eps)$-dependence on the error probability also in the advice-guided case, like we did for basic phase estimation (Theorem~\ref{thm: qpe lb})?

\paragraph{Acknowledgements.} We thank Jordi Weggemans for useful comments and for a pointer to~\cite{linlin&tong:groundstateprep}, and Han-Hsuan Lin for pointing out an error in an earlier version which we corrected here (see Footnote~\ref{note:hannote}).

\bibliographystyle{alphaurl}
\bibliography{bibo}

\newpage
\appendix

\section{Proof of Theorem~\ref{thm: gndeltat lower bound}}\label{app: adversary proof}

In this section we prove Theorem~\ref{thm: gndeltat lower bound} by a simple modification of the adversary method~\cite{Amb02}.

\begin{proof}[Proof of Theorem~\ref{thm: gndeltat lower bound}]
    Let $\mathcal{A}$ be a cost-$C$ algorithm solving $\for_{N,\delta,t}$.
Define $U_j = I - (1 - e^{i\delta})\ketbra{j}{j}$ for $j \in \cbra{1,\dots,N-1}$, and set $U=I$ if $j = 0$.
We will use the following advice states: $\ket{\alpha_0}=\ket{0}$ and $\ket{\alpha_j}=\gamma\ket{j} + \sqrt{1 - \gamma^2}\ket{0}$
for $j \in \cbra{1,\dots,N-1}$.
For $j \in \cbra{0,1,\dots,N-1}$ let $\ket{\psi_j^0}=\ket{\alpha_j}^{\otimes t}\ket{0}$ be the initial state, which includes $t$ copies of the advice state, and for $T\in\{1,\ldots,C\}$ let $\ket{\psi_j^T}$ be the state of the algorithm (on input $U = U_j$ with initial state $\ket{\psi_j^0}$) just before the $T$th application of $U$ or its inverse. 

Define the following  progress measure $P$ as a function of the timestep~$T$:
$$
P(T) = \sum_{j = 1}^{N-1} |\braket{\psi_0^T}{\psi_j^T}|.
$$
We have 
\begin{equation}\label{eqn: initial measure}
P(0) = \sum_{j = 1}^{N - 1}|\braket{\psi_0^0}{\psi_j^0}| =\sum_{j = 1}^{N - 1}(1 - \gamma^2)^{t/2} = (N-1)(1 - \gamma^2)^{t/2}.
\end{equation}
Since the output of $\mathcal{A}$ is different with high probability for $U = I$ on the one hand or for one of the other $U_j$ on the other hand, it is easy to show that $|\braket{\psi_0^C}{\psi_j^C}|$ is bounded below~1, say $\leq 0.99$. Thus,
\begin{equation}\label{eqn: final measure}
P(C) \leq 0.99 (N-1).
\end{equation}
The assumption $t=o(1/\gamma^2)$ implies $(1 - \gamma^2)^{t/2} - 0.99=\Omega(1)$, so we see that $P(C)$ is significantly smaller than $P(0)$. We now want to upper bound how much $P(T)$ can shrink in one step, in order to lower bound the number of steps.

For all $j\in\{0,1,\ldots,N-1\}$, define real amplitudes $\alpha_{jk}$ and normalized workspace states $\ket{w_{jk}}$ (depending on $T$, but we suppress this dependence in our notation) such that
\begin{align*}
\ket{\psi_j^T} & =\sum_{k = 0}^{N - 1} \alpha_{jk}\ket{k}\ket{w_{jk}}.
\end{align*}
Then
\begin{equation}\label{eqn:ipbefore}
    \braket{\psi_0^T}{\psi_j^T}=\sum_{k=0}^{N-1}\alpha_{0k}\alpha_{jk}\braket{w_{0k}}{w_{jk}} \mbox{~~~for all }j\in\{1,\ldots,N-1\}.
\end{equation}
We also have, after applying $U$:
\begin{align*}
I\ket{\psi_0^T} & = \hspace*{2.3em}\sum_{k = 0}^{N - 1}\hspace*{2.4em} \alpha_{0k}\ket{k}\ket{w_{0k}},\\
U_j\ket{\psi_j^T} & = \sum_{k\in\{0,\ldots,N-1\}\backslash\{j\}} \alpha_{jk}\ket{k}\ket{w_{jk}}~~+~~ e^{i\delta}\alpha_{jj}\ket{j}\ket{w_{jj}} \mbox{~~~for all }j\in\{1,\ldots,N-1\}.
\end{align*}
Using the fact that inner products are not changed by the input-independent unitary that follows the application of $U$, we have for all $j\in\{1,\ldots,N-1\}$
\begin{equation}\label{eqn:ipafter}
    \braket{\psi_0^{T+1}}{\psi_j^{T+1}}=\sum_{k\in\{0,\ldots,N-1\}\backslash\{j\}} \alpha_{0k}\alpha_{jk}\braket{w_{0k}}{w_{jk}}~~+~~e^{i\delta}\alpha_{0j}\alpha_{jj}\braket{w_{0j}}{w_{jj}}.
\end{equation}
Using Equations~\eqref{eqn:ipbefore} and~\eqref{eqn:ipafter},
\begin{align*}    |\braket{\psi_{0}^T}{\psi_j^T}| - |\braket{\psi_{0}^{T + 1}}{\psi_j^{T + 1}}| & \leq |\braket{\psi_{0}^T}{\psi_j^T} - \braket{\psi_{0}^{T + 1}}{\psi_j^{T + 1}}|\\
    & = |(1 - e^{i\delta})\alpha_{0j}\alpha_{jj}\braket{w_{0j}}{w_{jj}}|\\    & = |1 - e^{i\delta}|\cdot|\alpha_{0j}|\cdot|\alpha_{jj}|\cdot |\braket{w_{0j}}{w_{jj}}|\\
    & \leq |1 - e^{i\delta}|\cdot|\alpha_{0j}|,
\end{align*}
where the first inequality uses the triangle inequality and the last inequality holds because $|\alpha_{jj}|\cdot |\braket{w_{0j}}{w_{jj}}| \leq 1$ for all $j, T$.
Summing over all $j \in \cbra{1, \dots, N-1}$, using the Cauchy-Schwarz inequality and the fact that $\sum_{j=1}^{N-1}|\alpha_{0j}|^2\leq 1$, we can bound the change in the progress measure in one step by:
\begin{equation}\label{eqn: progress change}
P(T) - P(T+1)=\sum_{j=1}^{N-1} |\braket{\psi_{0}^T}{\psi_j^T}| - |\braket{\psi_{0}^{T + 1}}{\psi_j^{T + 1}}|\leq \sum_{j = 1}^{N - 1} |1 - e^{i\delta}|\cdot |\alpha_{0j}| \leq |1 - e^{i\delta}| \sqrt{N-1}.
\end{equation}
By Equations~\eqref{eqn: initial measure},~\eqref{eqn: final measure},~\eqref{eqn: progress change} and a telescoping sum, we have
\begin{equation}\label{eqn:advineq}
(N-1)((1 - \gamma^2)^{t/2} - 0.99) \leq P(0) - P(C) = \sum_{T = 0}^{C-1}(P(T) - P(T+1))\leq  C|1 - e^{i\delta}|\sqrt{N-1}.
\end{equation}
We have $|1 - e^{i\delta}|=2\sin(\delta/2)\leq \delta$ for all $\delta\in[0,\pi]$.
As already mentioned, the assumption $t=o(1/\gamma^2)$ implies $(1 - \gamma^2)^{t/2} - 0.99=\Omega(1)$.
Hence we obtain the desired lower bound $C=\Omega(\sqrt{N}/\delta)$ by rearranging Equation~\eqref{eqn:advineq}.
\end{proof}
\end{document}